\documentclass[11pt]{article}

\usepackage[T1]{fontenc}
\usepackage[utf8]{inputenc}
\usepackage{lmodern}
\usepackage[margin=1in]{geometry}
\usepackage{microtype}
\usepackage{amsmath,amssymb}
\usepackage{amsthm}
\usepackage{booktabs}
\usepackage{graphicx}
\usepackage{xcolor}
\usepackage{algorithm}
\usepackage[noend]{algpseudocode}
\usepackage{pgfplots}
\usepackage{pgfplotstable}
\usepackage{hyperref}
\usepackage[nameinlink,noabbrev]{cleveref}

\pgfplotsset{compat=1.18}
\hypersetup{
  colorlinks = true,
  linkcolor = blue!50!black,
  citecolor = blue!50!black,
  urlcolor = blue!50!black
}
\Crefname{theorem}{Theorem}{Theorems}
\Crefname{algorithm}{Algorithm}{Algorithms}
\makeatletter
\providecommand{\theHALG@line}{\thealgorithm.\arabic{ALG@line}}
\makeatother

\theoremstyle{plain}
\newtheorem{theorem}{Theorem}
\newtheorem{lemma}[theorem]{Lemma}

\theoremstyle{remark}
\newtheorem{remark}[theorem]{Remark}

\newcommand{\eqdist}{\mathrel{\overset{\mathrm{d}}{=}}}
\newcommand{\paperTitle}{Extended One-Liners for the Beta, Gamma, and Dirichlet Distributions with Shape Parameters Below One}
\title{\paperTitle}
\author{Dylan Greaves\\\texttt{dgreaves103@gmail.com}}
\date{}

\begin{document}
\maketitle
\begin{abstract}
We present an explicit deterministic transformation of a fixed number of i.i.d.\ uniform random variables with exact $\operatorname{Beta}(a,1-a)$ law for $0<a<1$, using only elementary operations (an ``extended one-liner'', see \cite{devroye1996oneline}).
As corollaries, the families $\operatorname{Beta}(a,b)$ with $\min(a,b)<1$, $\operatorname{Gamma}(c)$ with $c<1$, and $\operatorname{Dirichlet}(\alpha_1,\dots,\alpha_d)$ with $0<\alpha_i<1$, for fixed $d$, also have extended one-liners.
\end{abstract}

\section{Introduction}
\label{sec:intro}

Finding ways to sample from a parametric family $\{P_\theta\}_{\theta\in\Theta}$ by applying a simple deterministic transformation $T$ to a fixed number of i.i.d.\ uniform random variables
\[
T(U_1,\dots,U_k, \theta)\sim P_\theta, \qquad U_i\overset{\mathrm{i.i.d.}}{\sim}\operatorname{Unif}(0,1)
\]
has practical appeal.
Besides being easy to implement, such transformations have applications in machine learning through the reparameterization trick \cite{figurnov2018implicit,mohamed2020montecarlo}, in quasi-Monte Carlo \cite{practicalqmc},
and as sampling methods on GPUs, where performance is sensitive to looping and branching \cite{howes2007cuda}.

Although every real-valued random variable can be represented as a transformation of a single uniform random variable via the inversion method, the required quantile function often does not exist in closed-form and can be expensive to evaluate.
Nevertheless, many commonly used distributions with intractable quantile functions can be expressed as simple transformations of i.i.d.\ uniform random variables, including the Normal distribution via the Box--Muller formula, Student's \textit{t} distribution, and the symmetric Beta distribution \cite{devroye1996oneline}.
However, no such representation is known for several important families, perhaps most notably the general Gamma, Beta, and Dirichlet distributions \cite{devroye1996oneline,devroyejames2014stable,figurnov2018implicit}.

To make our notion of a simple transformation precise, we adopt Devroye's ``one-liner'' framework \cite{devroye1996oneline}.
We call a transformation a \textit{simple one-liner} if it can be represented as a fixed expression tree whose leaves are $U_1,\dots,U_k$, $\theta$, and constants, and whose internal nodes are operators from some admissible set $\mathcal{F}$, e.g.
\[
\mathcal{F} = \{+, -, \times, /, \bmod, \operatorname{round}, |\cdot|, \operatorname{sign},
\lfloor \cdot \rfloor, \lceil \cdot \rceil, \sin, \cos, \exp,
\log, \tan, \arctan\}.
\]
These transformations are called ``one-liners'' as they can be written in a single line of code.
If we allow variable reuse (i.e. a directed acyclic graph instead of a tree), we call the transformation an \textit{extended one-liner}.
Several additional operators can be expressed by combining operators in $\mathcal{F}$, for example
\[
  \mathbf{1}_{\{x>y\}} = \lceil\operatorname{sign}(x-y)/2\rceil,
  \qquad
  \max(x,y) = x + (y-x)\mathbf{1}_{\{y>x\}},
  \qquad
  x^c = \exp(c\log x)\quad (x>0).
\]
Likewise, a random variable with mixture distribution $wF+(1-w)G$ can be written as
\[
  \mathbf{1}_{\{U_{k+1}< w\}}T_F(U_1,\dots,U_k,\theta)
  + \mathbf{1}_{\{U_{k+1}\geq w\}}T_G(U_1,\dots,U_k,\tilde\theta),
\]
provided both $F$ and $G$ have one-liners.

In what follows, we construct an extended one-liner for the $\operatorname{Beta}(a,1-a)$ family with $0<a<1$ by applying a correction step to the proposal distribution used in J\"ohnk's rejection method \cite{johnk1964beta} for generating $\operatorname{Beta}(a,b)$ random variates.
We then extend this to the subfamilies $\operatorname{Beta}(a,b)$ with $\min(a,b)<1$, $\operatorname{Gamma}(c)$\footnote{We suppress the Gamma scale parameter, writing $\operatorname{Gamma}(c)$ for $\operatorname{Gamma}(c,1)$. A one-liner for $\operatorname{Gamma}(c)$ immediately yields one for $\operatorname{Gamma}(c,\theta)$, since if $X\sim\operatorname{Gamma}(c)$, then $\theta X\sim\operatorname{Gamma}(c,\theta)$.} with $c<1$, and $\operatorname{Dirichlet}(\alpha_1,\dots,\alpha_d)$ with $0<\alpha_i<1$ (for fixed $d$).
To the best of the author's knowledge, these are the first known extended one-liners for these subfamilies, apart from overlap with a few previously known cases ($\operatorname{Beta}(a,1)$ for $a>0$ via $U^{\frac{1}{a}}$, the symmetric Beta family $\operatorname{Beta}(a,a)$ for $a>0$, the half-Beta family $\operatorname{Beta}(\frac{1}{2}, b)$ for $b>0$, $\operatorname{Gamma}(1)$, $\operatorname{Gamma}(\frac{1}{2})$, see \cite{devroye1996oneline}).

\section{The \texorpdfstring{$\operatorname{Beta}(a,1-a)$}{Beta(a,1-a)} Generator}
\label{sec:algorithm}
Let $0<a<1$, and let $U_1,U_2 \overset{\mathrm{i.i.d.}}{\sim} \operatorname{Unif}(0,1)$.
Following J{\"o}hnk, let
\[
  P = \frac{U_1^{1/a}}{U_1^{1/a}+U_2^{1/(1-a)}}.
\]
Unlike J{\"o}hnk's method, which uses a rejection step, we apply a single conditional mixture step: given $P=p$, let
\begin{equation}
  B \mid P=p \sim \alpha_a(p)\operatorname{Unif}(0,p)
  + \bigl(1-\alpha_a(p)\bigr)\operatorname{Unif}(p,1),
\label{eq:mixture-def}
\end{equation}
where 
\begin{equation}
  \alpha_a(p) = p + \frac{\sin(\pi a)}{\pi a(1-a)}(1-a-p)\max(p,1-p).
\label{eq:alpha-def}
\end{equation}
The mixture weight $\alpha_a(p)$ is chosen based on a derivative-matching argument.
We prove in Theorem~\ref{thm:beta-law} that $B$ is distributed exactly $\operatorname{Beta}(a,1-a)$.
The pseudocode below implements this procedure:
\begin{algorithm}[H]
  \caption{\texorpdfstring{$\operatorname{Beta}(a,1-a)$}{Beta(a,1-a)} generator, parameter $0<a<1$}
  \label{alg:beta-core}
  \begin{algorithmic}[1]
    \State $U_1,U_2,U_3 \overset{\mathrm{i.i.d.}}{\sim} \operatorname{Unif}(0,1)$
    \State $P \gets \frac{U_1^{1/a}}{U_1^{1/a}+U_2^{1/(1-a)}}$
    \State $A \gets P + \frac{\sin(\pi a)}{\pi a(1-a)}(1-a-P)\max(P,1-P)$ \Comment{$A = \alpha_a(P)$}
    \State \parbox[t]{0.82\linewidth}{$B \gets \frac{P}{A}U_3$ if $U_3 \le A$, else $B \gets P + \frac{1-P}{1-A}(U_3-A)$}
    \State \textbf{return} $B$
  \end{algorithmic}
\end{algorithm}
We note that \Cref{alg:beta-core} can be modified to require only two uniform inputs; see Appendix~\ref{sec:two-uniform-variant}.

\begin{theorem}
\label{thm:beta-law}
Let $B$ be defined as above. Then
\[
  B \eqdist \operatorname{Beta}(a,1-a), \qquad 0<a<1.
\]
\end{theorem}
Before proving the theorem, we first verify that $\alpha_a(p)$ is a valid mixture weight for all $0<a<1$ and $0<p<1$:
\begin{lemma}
\label{lem:alpha-bounds}
For $0<a<1$ and $0<p<1$, define
\[
  \alpha_a(p) = p + \dfrac{\sin(\pi a)}{\pi a(1-a)}(1-a-p)\max(p,1-p).
\]
Then $0 < \alpha_a(p) < 1$.
\end{lemma}

\begin{proof}
Set
\[
  \kappa(a) = \frac{\sin(\pi a)}{\pi a(1-a)}>0
  \qquad\text{and}\qquad
  M(p) = \max(p,1-p).
\]
Since $\kappa(a)=\kappa(1-a)$ and $M(p)=M(1-p)$, we have
\[
  \begin{aligned}
    \alpha_{1-a}(1-p)
      &= 1-p + \kappa(1-a)(p+a-1)M(1-p) \\
      &= 1-p-\kappa(a)(1-a-p)M(p) \\
      &= 1-\alpha_a(p).
  \end{aligned}
\]
By symmetry, it suffices to show the inequality for $p \geq \frac{1}{2}$, where
\[
  \alpha_a(p) = p + \kappa(a)(1-a-p)p.
\]
From the inequalities $\sin(\pi a) \leq \pi a$ and $\sin(\pi a) \leq \pi (1 - a)$, we have $\kappa(a)\leq \frac{1}{1-a}$ and $\kappa(a) \leq \frac{1}{a}$, respectively.
We split into two cases:
\begin{enumerate}
\item If $a + p \leq 1$, then
\[
  \alpha_a(p) = p + \kappa(a)(1-a-p)p\geq p > 0.
\]
For the upper bound, using $\kappa(a) \leq \frac{1}{1-a}\leq\frac{1}{p}$, we have
\[
  \alpha_a(p) = p + \kappa(a)(1-a-p)p \leq p + 1 - a-p =1-a< 1.
\]
\item If $a + p > 1$, then
\[
  \alpha_a(p) = p + \kappa(a)(1-a-p)p < p < 1.
\]
For the lower bound, using $\kappa(a) \leq \frac{1}{a}$ we have
\[
  \alpha_a(p) = p - \kappa(a)(a+p-1)p \geq p-\frac{a+p-1}{a}p =\frac{p(1-p)}{a} > 0.\qedhere
\]

\end{enumerate}
\end{proof}

\begin{proof}[Proof of \Cref{thm:beta-law}]
Let
\[
  X=U_1^{\frac{1}{a}},
  \qquad
  Y=U_2^{\frac{1}{1-a}},
  \qquad
  P=\frac{X}{X+Y},
  \qquad
  S=X+Y.
\]
Using the change of variables $x=ps$, $y=(1-p)s$, with Jacobian $s$, the joint density of $(P,S)$ is
\[
  \begin{aligned}
    f_{P,S}(p,s)
      &= a(1-a)s(ps)^{a-1}((1-p)s)^{-a} \\
      &= a(1-a)p^{a-1}(1-p)^{-a},
  \end{aligned}
\]
where $0<p<1$ and $0<s<\frac{1}{\max(p,1-p)}$, hence the marginal density of $P$ is
\begin{equation}
f_P(p)= a(1-a)\frac{p^{a-1}(1-p)^{-a}}{\max(p,1-p)}
\qquad 0<p<1.
\label{eq:fp-marginal}
\end{equation}
Since the conditional density of $B$ given $P=p$ is
\[
 f_{B \mid P}(b \mid p) = \frac{\alpha_a(p)}{p}\mathbf{1}_{\{0<b<p\}} + \frac{1-\alpha_a(p)}{1-p}\mathbf{1}_{\{p<b<1\}},
\]
the marginal density of $B$ is 
\begin{equation}
f_B(b)= \int_0^1 f_{B \mid P}(b \mid p)f_P(p)\,\mathrm{d}p = \int_b^1 \frac{\alpha_a(p)}{p}f_P(p)\,\mathrm{d}p + \int_0^b\frac{1-\alpha_a(p)}{1-p}f_P(p)\,\mathrm{d}p
\qquad 0<b<1.
\label{eq:fb-marginal}
\end{equation}
We want to show that $f_B(b)$ matches the target $\operatorname{Beta}(a,1-a)$ density
\begin{equation}
g(b) = \frac{1}{\Gamma(a)\Gamma(1-a)}b^{a-1}(1-b)^{-a} = \frac{\sin(\pi a)}{\pi}b^{a-1}(1-b)^{-a},
\label{eq:beta-density}
\end{equation}
where the second equality follows from Euler's reflection formula.
Differentiating $g$, we have
\begin{equation}
  g'(b)= \left(\frac{a-1}{b}+\frac{a}{1-b}\right)g(b)= \left(\frac{a+b-1}{b(1-b)}\right)g(b).
\label{eq:beta-deriv}
\end{equation}
	Substituting in the definition of $\alpha_a(p)$
\[
  \alpha_a(p) = p + \frac{\sin(\pi a)}{\pi a(1-a)}(1-a-p)\max(p,1-p),
\]
we have
\[
 \frac{\alpha_a(p)}{p}f_P(p) = f_P(p)-(1-p)g'(p),
\]
and 
\[
 \frac{1-\alpha_a(p)}{1-p}f_P(p) = f_P(p)+pg'(p),
\]
hence
\[
\begin{aligned}
f_B(b) &= \int_b^1 \frac{\alpha_a(p)}{p}f_P(p)\,\mathrm{d}p + \int_0^b\frac{1-\alpha_a(p)}{1-p}f_
P(p)\,\mathrm{d}p \\
&= \int_b^1\left(f_P(p) - (1-p)g'(p)\right)\,\mathrm{d}p + \int_0^b\left(f_P(p) + pg'(p)\right)\,\mathrm{d}p \\
&= \int_0^1 f_P(p)\,\mathrm{d}p + \int_0^b pg'(p)\,\mathrm{d}p -  \int_b^1 (1-p)g'(p)\,\mathrm{d}p.\\
&=  1 + \int_0^b pg'(p)\,\mathrm{d}p -  \int_b^1 (1-p)g'(p)\,\mathrm{d}p.
\end{aligned}
\]
Using integration by parts,
\[
\begin{aligned}
\int_\epsilon^b pg'(p)\,\mathrm{d}p &= [pg(p)]_\epsilon^b -\int_\epsilon^b g(p)\,\mathrm{d}p\\
&\rightarrow bg(b)-\int_0^bg(p)\,\mathrm{d}p\qquad\text{as }\epsilon\rightarrow 0^+,
\end{aligned}
\]
as $\epsilon g(\epsilon)=\frac{\sin(\pi a)}{\pi}\left(\frac{\epsilon}{1-\epsilon}\right)^a\rightarrow 0,$ and
\[
\begin{aligned}
\int_b^{1-\epsilon} (1-p)g'(p)\,\mathrm{d}p &= [(1-p)g(p)]_b^{1-\epsilon} +\int_b^{1-\epsilon} g(p)\,\mathrm{d}p\\
&\rightarrow -(1-b)g(b)+\int_b^1 g(p)\,\mathrm{d}p\qquad\text{as }\epsilon\rightarrow 0^+,
\end{aligned}
\]
as $\epsilon g(1-\epsilon)=\frac{\sin(\pi a)}{\pi}\left(\frac{\epsilon}{1-\epsilon}\right)^{1-a}\rightarrow 0.$

Thus
\begin{align*}
f_B(b) &=  1 + \int_0^b pg'(p)\,\mathrm{d}p -  \int_b^1 (1-p)g'(p)\,\mathrm{d}p\\
&= 1 + bg(b)-\int_0^bg(p)\,\mathrm{d}p + (1-b)g(b)-\int_b^1g(p)\,\mathrm{d}p\\
&= 1 + g(b)-\int_0^1g(p)\,\mathrm{d}p\\
&= g(b). \qedhere
\end{align*}
\end{proof}
\begin{remark}[Motivation for the mixture weight]
Suppose $P\sim f_P$ and $B\mid P=p$ is generated via a mixture as in \eqref{eq:mixture-def}.
Applying Leibniz's rule to \eqref{eq:fb-marginal}, when justified, gives
\[
f_B'(b)= -\frac{\alpha(b)}{b}f_P(b)
+ \frac{1-\alpha(b)}{1-b}f_P(b),
\qquad 0<b<1.
\]
If $g = f_B$ on $(0,1)$, then rearranging yields, at points where $f_P(b)>0$ and $g'(b)$ exists,
\begin{equation}
\alpha(b)= b - \frac{b(1-b)g'(b)}{f_P(b)}.
\label{eq:alpha-force}
\end{equation}
Thus, under these conditions, the mixture weight $\alpha$ is essentially determined by the source and target densities.
In this particular case, substituting \eqref{eq:fp-marginal} and \eqref{eq:beta-deriv} into \eqref{eq:alpha-force} yields 
\[
\alpha(b)
= b + \frac{\sin(\pi a)}{\pi a(1-a)}(1-a-b)\max(b,1-b).
\]
\end{remark}

\section{Further Constructions}
\label{sec:consequences}
Theorem~\ref{thm:beta-law} yields a few additional extended one-liners, which we record below:

\paragraph{Gamma.}
By \cite[Section IX.3.5]{devroye1986nonuniform}, if $0<a<1$, let
$B \sim \operatorname{Beta}(a,1-a)$ and $U\sim\operatorname{Unif}(0,1)$ be
independent, and set $E=-\log U \eqdist \operatorname{Exp}(1)$. Then
\[
  EB \eqdist \operatorname{Gamma}(a).
\]
Thus $\operatorname{Gamma}(a)$ with $0<a<1$ has an extended one-liner.

\paragraph{Beta.}
If $0<a<1$ and $b>0$, let $B \sim \operatorname{Beta}(a,1-a)$ and
$U\sim\operatorname{Unif}(0,1)$ be independent, and set
$V = 1-U^{1/b} \eqdist \operatorname{Beta}(1,b)$. Then
\[
  \frac{BV}{1-(1-B)V} \eqdist \operatorname{Beta}(a,b).
\]
By symmetry, $\operatorname{Beta}(a,b) \eqdist 1-\operatorname{Beta}(b,a)$. Thus $\operatorname{Beta}(a,b)$ with $\min(a,b)<1$ has an extended one-liner.

\paragraph{Dirichlet.}
By the previous constructions and \cite[Theorems XI.4.1, XI.4.2]{devroye1986nonuniform},
$\operatorname{Dirichlet}(\alpha_1,\dots,\alpha_d)$ with $0<\alpha_i<1$ for each $i$ (for fixed $d$) has an extended one-liner.

\section{Numerical Checks}
\label{sec:evaluation}
For $a\in\{0.05,0.10,\dots,0.95\}$, we simulate $N=10^8$ draws from \Cref{alg:beta-core}, and compare the empirical moments $\widehat{m}_k=\frac{1}{N}\sum_i X_i^k$ with the theoretical $B\sim\operatorname{Beta}(a,1-a)$ moments $\mathbb{E}[B^k]=\frac{a(a+1)\dots(a+k-1)}{k!}$ for $k=1,2,3.$
The table reports the empirical moment, the theoretical moment, the Monte Carlo standard error (MCSE), and $z=(\widehat{m}_k-\mathbb{E}[B^k])/\operatorname{MCSE}$ for each value of $a$.
\begin{table}[H]
  \centering
  \setlength{\tabcolsep}{4pt}
  \small
  \caption{
  \texorpdfstring{$\operatorname{Beta}(a,1-a)$}{Beta(a,1-a)} sampler moment diagnostics.}
  \label{tab:ks-moments}
  \resizebox{\textwidth}{!}{
  \begin{filecontents*}{ks_moments.csv}
a,emp_mom_1,the_mom_1,mcse_mom_1,z_mom_1,emp_mom_2,the_mom_2,mcse_mom_2,z_mom_2,emp_mom_3,the_mom_3,mcse_mom_3,z_mom_3,emp_mom_4,the_mom_4,mcse_mom_4,z_mom_4
0.05,0.04999803464995114,0.05,1.5411035007422442e-05,-0.12752875117848514,0.026245525520320438,0.026250000000000002,1.1396614080506541e-05,-0.3926148282249629,0.017932297671006726,0.0179375,9.48829552412471e-06,-0.5482890978728026,0.013671883707568468,0.013677343749999998,8.315632573108457e-06,-0.6565997695937519
0.1,0.10000442048349835,0.1,2.1213203435596428e-05,0.20838359052013108,0.0549981367879921,0.05500000000000001,1.6374522893812816e-05,-0.1137872547488411,0.0384946385247252,0.038500000000000006,1.3897657176661108e-05,-0.38578266873726735,0.0298302311053767,0.029837500000000006,1.2326301889912075e-05,-0.5897060357783503
0.15,0.14998454581587678,0.15,2.5248762345905197e-05,-0.6120768975324087,0.08624854905200334,0.08625,2.030721084984346e-05,-0.07144989074975046,0.061818203181299325,0.06181249999999999,1.7566378630871103e-05,0.32466460043788964,0.04868686240132398,0.04867734374999999,1.576734583562869e-05,0.603693952249021
0.2,0.199955190057641,0.2,2.82842712474619e-05,-1.5842707053314726,0.1199738870237427,0.12,2.3664319132398465e-05,-1.1034746493737546,0.08798249584902394,0.08800000000000001,2.085837961108197e-05,-0.839190354305692,0.07038684909422543,0.0704,1.8946864648273605e-05,-0.6940940371244051
0.25,0.2500271670995433,0.25,3.061862178478972e-05,0.8872737556335185,0.1562726691548585,0.15625,2.6608416196760002e-05,0.8519543099017826,0.1172069576699974,0.1171875,2.3892076067770356e-05,0.8143984617415919,0.09523196982694698,0.09521484375,2.196248136811614e-05,0.7797878873488953
0.3,0.3000031465413936,0.3,3.24037034920393e-05,0.09710437556556022,0.1950054495766159,0.195,2.9208303613869806e-05,0.18657627940079533,0.14950784451905627,0.1495,2.67105105716832e-05,0.2936866008317077,0.12334654065925618,0.12333749999999999,2.484656879291887e-05,0.3638594661315837
0.35,0.3500033458940345,0.35,3.3726843908080104e-05,0.09920566666777371,0.23624739485198748,0.23625,3.149218653094764e-05,-0.08272363082657576,0.18505716475224163,0.1850625,2.9323240336431526e-05,-0.18194605020291732,0.1549826598188398,0.15498984375,2.760165732670388e-05,-0.26027173206183446
0.4,0.399948208124054,0.4,3.4641016151377533e-05,-1.4951026759631039,0.27993293228610006,0.27999999999999997,3.3466401061363015e-05,-2.0040312603958217,0.22393056895317232,0.22399999999999995,3.172078183147446e-05,-2.1888189010127705,0.19033266121307746,0.19039999999999996,3.021238951158945e-05,-2.228846774819595
0.45,0.4500021686914473,0.45,3.5178118198675715e-05,0.06164887601618639,0.3262391987520205,0.32625,3.5123109379723204e-05,-0.3075253919777001,0.2664229767170883,0.2664375,3.3880245471411875e-05,-0.4286652209747808,0.22978670463656925,0.22980234375,3.264990928234916e-05,-0.4789940852669811
0.5,0.5000114422480805,0.5,3.535533905932738e-05,0.323635648388855,0.37501838997736214,0.375,3.6443449342783126e-05,0.5046168157456523,0.3125236052073931,0.3125,3.576725981956124e-05,0.6599668946452261,0.27346521321389206,0.2734375,3.487299081354007e-05,0.794689908882015
0.55,0.5500327524811423,0.55,3.517811819867573e-05,0.9310469922589371,0.4262918266670655,0.4262500000000001,3.73982995936981e-05,1.1184109309738135,0.36235583987968145,0.36231250000000004,3.733581895382937e-05,1.1608123484584558,0.32159591476139293,0.32155234375,3.682723494858295e-05,1.1831192717492678
0.6,0.5999947107467138,0.6,3.464101615137755e-05,-0.15268759042807084,0.479989156805835,0.48,3.794733192202056e-05,-0.2857432556070639,0.41598688844433956,0.416,3.852645844092083e-05,-0.34032600428411663,0.37438583514914797,0.3744,3.844257639649039e-05,-0.3684677818142243
0.65,0.6499589027844938,0.65,3.37268439080801e-05,-1.2185313164278495,0.5362008641797537,0.53625,3.803626969748742e-05,-1.2918149081689407,0.4736344855400109,0.4736875,3.926247515221277e-05,-1.350257715122092,0.43218517668602496,0.43223984374999996,3.962868119327576e-05,-1.3794822923422487
0.7,0.7000174995856203,0.7,3.240370349203931e-05,0.540048936834059,0.5950260988579555,0.595,3.759155490266398e-05,0.6942744992365385,0.5355315970730351,0.5355,3.944314103871548e-05,0.8010790267467204,0.49537294487022454,0.4953375,4.02670036230814e-05,0.8802460336094479
0.75,0.7500128596526034,0.75,3.061862178478972e-05,0.41999449530532723,0.6562619633710436,0.65625,3.6510379517337256e-05,0.3276704104899145,0.6015724378408394,0.6015625,3.8931420698172896e-05,0.25526530142334103,0.5639725014794592,0.56396484375,4.019668443301448e-05,0.19050649493170718
0.8,0.8000029486053809,0.8,2.8284271247461896e-05,0.10424894299243313,0.7200125158682004,0.7200000000000001,3.464101615137753e-05,0.36130199372868177,0.6720185872520634,0.672,3.753014788140328e-05,0.4952618924423361,0.6384224431977075,0.6384,3.918752454544686e-05,0.5727128204148781
0.85,0.8500058119729407,0.85,2.524876234590521e-05,0.2301884290830303,0.7862383188565655,0.78625,3.1739294454981186e-05,-0.3680341241074562,0.7469124265131445,0.7469375,3.492991646788071e-05,-0.7178226973029954,0.7188926421376355,0.71892734375,3.688139875855871e-05,-0.9408974044528456
0.9,0.9000035845806996,0.9,2.1213203435596407e-05,0.16897875469319265,0.8550009684455879,0.855,2.7351873793215714e-05,0.03540691929434541,0.8264989117599932,0.8265,3.0570653820289833e-05,-0.03559753786006325,0.805835028825498,0.8058375,3.263780811085897e-05,-0.07571508765644225
0.95,0.9499917462042943,0.95,1.5411035007422422e-05,-0.5355769876394645,0.9262428670493627,0.9262499999999999,2.036867233032142e-05,-0.35019222272224687,0.9108070430588381,0.9108124999999999,2.3115662346191306e-05,-0.23607115730036843,0.8994231417967722,0.89942734375,2.4951273322286457e-05,-0.16840636441821552
\end{filecontents*}
\pgfplotstabletypeset[
    col sep=comma,
    columns={
      a,
      emp_mom_1, the_mom_1, mcse_mom_1, z_mom_1,
      emp_mom_2, the_mom_2, mcse_mom_2, z_mom_2,
      emp_mom_3, the_mom_3, mcse_mom_3, z_mom_3
    },
    columns/a/.style={column name={$a$}, fixed, fixed zerofill, precision=2},
    columns/emp_mom_1/.style={column name={Emp.}, fixed, fixed zerofill, precision=6},
    columns/the_mom_1/.style={column name={Theo.}, fixed, fixed zerofill, precision=6},
    columns/mcse_mom_1/.style={column name={MCSE}, sci, sci zerofill, precision=2},
    columns/z_mom_1/.style={column name={$z$}, fixed, fixed zerofill, precision=2},
    columns/emp_mom_2/.style={
      column name={Emp.},
      fixed,
      fixed zerofill,
      precision=6
    },
    columns/the_mom_2/.style={column name={Theo.}, fixed, fixed zerofill, precision=6},
    columns/mcse_mom_2/.style={column name={MCSE}, sci, sci zerofill, precision=2},
    columns/z_mom_2/.style={column name={$z$}, fixed, fixed zerofill, precision=2},
    columns/emp_mom_3/.style={
      column name={Emp.},
      fixed,
      fixed zerofill,
      precision=6
    },
    columns/the_mom_3/.style={column name={Theo.}, fixed, fixed zerofill, precision=6},
    columns/mcse_mom_3/.style={column name={MCSE}, sci, sci zerofill, precision=2},
    columns/z_mom_3/.style={column name={$z$}, fixed, fixed zerofill, precision=2},
    every head row/.style={
      before row={
        \toprule
        & \multicolumn{4}{c}{$k=1$}
        & \multicolumn{4}{c}{$k=2$}
        & \multicolumn{4}{c}{$k=3$}\\
        \cmidrule(lr){2-5}\cmidrule(lr){6-9}\cmidrule(lr){10-13}
      },
      after row=\midrule
    },
    every last row/.style={after row=\bottomrule},
  ]{ks_moments.csv}
  }
\end{table}

\section*{Acknowledgements}
The results in this paper stemmed from significant interaction with ChatGPT-5.4 Pro, which generated the construction and main proof ideas.
Harmonic's Aristotle \cite{achim2025aristotleimolevelautomatedtheorem} was used to independently verify the main result in Lean.
The author verified the correctness of the results, and takes responsibility for the exposition and any errors.
Codex was used for simulations, optimization, and typesetting.
The author thanks Luc Devroye for helpful correspondence, and Art Owen for introducing this problem in his Monte Carlo course.

\nocite{johnk1964beta,devroye1986nonuniform}
\bibliographystyle{plain}
\bibliography{refs}

\appendix
\section{Appendix}
\label{sec:appendix-details}
\subsection{Two Uniform Variant}
\label{sec:two-uniform-variant}

The ratio $P=\frac{U_1^{1/a}}{U_1^{1/a} + U_2^{1/(1-a)}}$ has a simple closed-form CDF:
\[
  F_P(p)=
  \begin{cases}
    (1-a)\left(\dfrac{p}{1-p}\right)^a, & 0<p\le \tfrac12, \\[1ex]
    1-a\left(\dfrac{1-p}{p}\right)^{1-a}, & \tfrac12 \le p<1.
  \end{cases}
\]
Using the inversion method, we can generate $F_P^{-1}(U) \eqdist P$ from a single uniform $U$. 
We can therefore modify \Cref{alg:beta-core} to construct a $\operatorname{Beta}(a,1-a)$ extended one-liner that requires only two uniforms instead of three:
\begin{algorithm}[H]
  \caption{$\operatorname{Beta}(a,1-a)$ generator, $0<a<1$}
  \label{alg:beta-compact}
  \begin{algorithmic}[1]
    \State $b \gets 1-a$
    \State $\kappa \gets \frac{\sin(\pi a)}{\pi a(1-a)}$
    \State $U_1,U_2 \overset{\mathrm{i.i.d.}}{\sim} \operatorname{Unif}(0,1)$
    \If{$U_1 \le b$}
      \State $P \gets \frac{(U_1/b)^{1/a}}{1+(U_1/b)^{1/a}}$
      \State $A \gets P+\kappa(b-P)(1-P)$
    \Else
      \State $P \gets \frac{1}{1+\left((1-U_1)/a\right)^{1/b}}$
      \State $A \gets P+\kappa(1-a-P)P$
    \EndIf
    \State \parbox[t]{0.82\linewidth}{$B \gets \frac{P}{A}U_2$ if $U_2 \le A$, else $B \gets P+\frac{1-P}{1-A}(U_2-A)$}
    \State \textbf{return} $B$
  \end{algorithmic}
\end{algorithm}

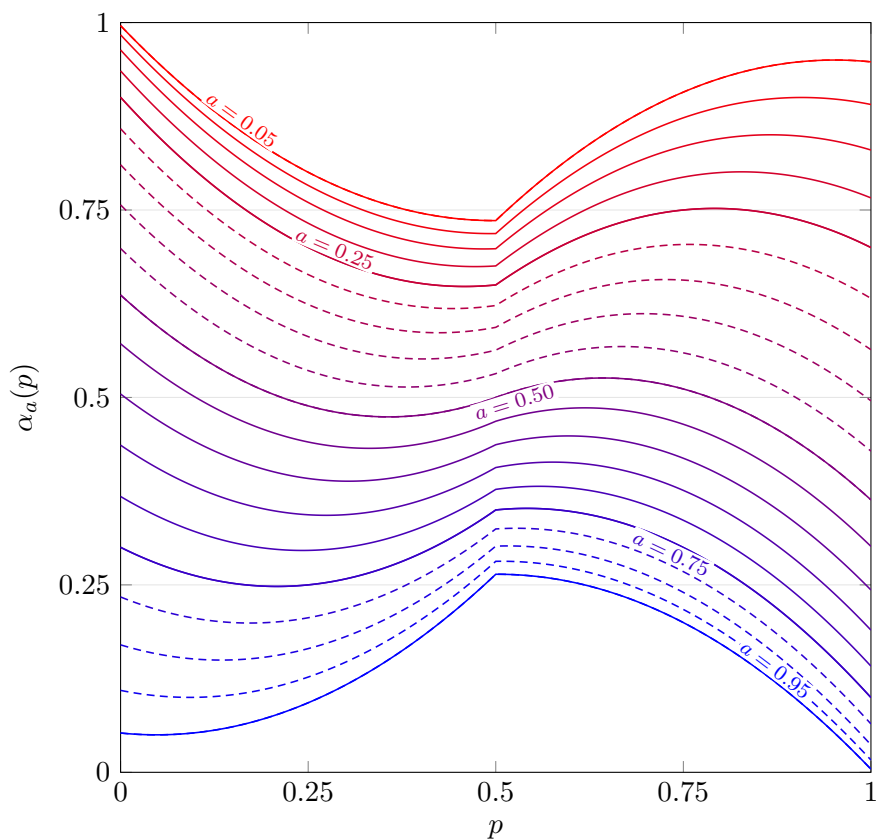
\begin{figure}[H]
  \centering
  \def\AlphaAValues{0.05,0.10,0.15,0.20,0.25,0.30,0.35,0.40,0.45,0.50,0.55,0.60,0.65,0.70,0.75,0.80,0.85,0.90,0.95}
\def\AlphaExpr#1{x + (sin(deg(pi*#1))/(pi*#1*(1-#1)))*(1-#1-x)*max(x,1-x)}

\begin{tikzpicture}
  \begin{axis}[
    width=0.70\linewidth,
    height=0.70\linewidth,
    xlabel={$p$},
    ylabel={$\alpha_a(p)$},
    ymin=0,
    ymax=1,
    xmin=0,
    xmax=1,
    xtick={0,0.25,0.5,0.75,1},
    ytick={0,0.25,0.5,0.75,1},
    ymajorgrids=true,
    grid style={gray!20},
    samples=201,
    domain=0:1
  ]
    \foreach \a [count=\i from 0] in \AlphaAValues {
      \pgfmathtruncatemacro{\mix}{round(100*\i/18)}
      \edef\AlphaPlotCommand{\noexpand\addplot+[no marks, semithick, blue!\mix!red]{\AlphaExpr{\a}};}
      \AlphaPlotCommand
    }

    \addplot+[no marks, semithick, blue!0!red]
      {\AlphaExpr{0.05}}
      node[pos=0.18, above, sloped, font=\scriptsize, fill=white, inner sep=1pt] {$a=0.05$};

    \addplot+[no marks, semithick, blue!22!red]
      {\AlphaExpr{0.25}}
      node[pos=0.32, above, sloped, font=\scriptsize, fill=white, inner sep=1pt] {$a=0.25$};

    \addplot+[no marks, semithick, blue!50!red]
      {\AlphaExpr{0.50}}
      node[pos=0.52, below, sloped, font=\scriptsize, fill=white, inner sep=1pt] {$a=0.50$};

    \addplot+[no marks, semithick, blue!78!red]
      {\AlphaExpr{0.75}}
      node[pos=0.70, below, sloped, font=\scriptsize, fill=white, inner sep=1pt] {$a=0.75$};

    \addplot+[no marks, semithick, blue!100!red]
      {\AlphaExpr{0.95}}
      node[pos=0.84, above, sloped, font=\scriptsize, fill=white, inner sep=1pt] {$a=0.95$};
  \end{axis}
\end{tikzpicture}
  \caption{The function $\alpha_a(p)$ on $(0,1)$ for
  $a\in\{0.05,0.10,\dots,0.95\}$.}
  \label{fig:alpha-profile}
\end{figure}

\end{document}